\documentclass[10pt, two column, twoside]{IEEEtran}
\usepackage{amsmath}
\usepackage{epsfig,amssymb,amsfonts}
\usepackage{amsbsy,bbm}
\usepackage[utf8]{inputenc}%latin1
\usepackage{verbatim}
\usepackage{graphicx}
\usepackage{theorem}
\usepackage{tabularx}
\usepackage{array}
\usepackage{multicol,float}
\usepackage{hhline}
 \usepackage{epstopdf}
  \usepackage{times}			% ajout times le 30 mai 2003
\usepackage[T1]{fontenc}

\def\Gr#1{\mathbf #1}
\def\RR{\mathbb{R}}
\def\NN{\mathcal{N}}
\def\E{\mathcal{E}}
\def\CE{\mathcal{CE}}

\def\GCN{\mathcal{GCN}}

\def\WW{\mathcal{W}}
\def\CW{\mathcal{CW}}
\def\gg{\mathbf{g}}
\def\kg{\mathbf{k}}

\def\pg{\mathbf{p}}

\def\sg{\mathbf{s}}
\def\tg{\mathbf{t}}
\def\ug{\mathbf{u}}
\def\vg{\mathbf{v}}
\def\wg{\mathbf{w}}
\def\xg{\mathbf{x}}
\def\yg{\mathbf{y}}
\def\zg{\mathbf{z}}

\def\alphag{\boldsymbol{\alpha}}

\def\mug{\boldsymbol{\mu}}

\def\Lg{\boldsymbol{\Lambda}}
\def\Pig{\boldsymbol{\Pi}}
\def\Sig{\boldsymbol{\Sigma}}
\def\Omg{\boldsymbol{\Omega}}
\def\Ag{\mathbf{A}}
\def\Bg{\mathbf{B}}
\def\Cg{\mathbf{C}}

\def\Ig{\mathbf{I}}
\def\Jg{\mathbf{J}}
\def\Kg{\mathbf{K}}
\def\Mg{\mathbf{M}}

\def\Rg{\mathbf{R}}
\def\Sg{\mathbf{S}}

\def\Trg{\mathbf{T_r}}
\def\Vg{\mathbf{V}}
\def\Wg{\mathbf{W}}

\newcommand\Mc{\widehat{\Gr{M}}}
\newcommand\Rc{\widehat{\Gr{R}}}

\newcommand\thetac{\widehat{\Gr{\theta}}}

\newcommand\alphagt{\overset{\sim}{\boldsymbol{\alpha}}}

\newtheorem{Thm}{Theorem}[section]
\newtheorem{Lem}{Lemma}[section]

\newtheorem{proof}{Proof}[section]

\newcommand\be{\left(}
\newcommand\ba{\left\{}
\newcommand\en{\right)}

\newcommand\disp{\displaystyle}
\newcommand\dsum{\displaystyle \sum}

\def\vec{\textup{vec}}

\begin{document}
%\ninept
\title{Asymptotic properties of robust complex covariance matrix estimates}

% Les auteurs et leurs caract?ristiques:
\author{Mélanie Mahot~\IEEEmembership{Student,~IEEE}\thanks{M.Mahot is with SONDRA, Supelec, Plateau du Moulon, 3 rue Joliot-Curie, F-91190 Gif-sur-Yvette, France (e-mail: melanie.mahot@supelec.fr)}, Frédéric Pascal~\IEEEmembership{Member,~IEEE}\thanks{F. Pascal is with SONDRA, Supelec, Plateau du Moulon, 3 rue Joliot-Curie, F-91190 Gif-sur-Yvette, France (e-mail: frederic.pascal@supelec.fr)}, Philippe Forster~\IEEEmembership{Member,~IEEE}\thanks{P. Forster is with SATIE, ENS Cachan, CNRS, UniverSud, 61, Av. du Pdt Wilson, F-94230 Cachan, France (e-mail:philippe.forster@satie.ens-cachan.fr)}, Jean-Philippe Ovarlez~\IEEEmembership{Member,~IEEE}\thanks{J.-P. Ovarlez is with ONERA, DEMR/TSI, Chemin de la Huni\`ere, F-91120 Palaiseau,  (e-mail:ovarlez@onera.fr)}}

\maketitle
\begin{abstract}
In many statistical signal processing applications, the estimation of nuisance parameters and parameters of interest is strongly linked to the resulting performance. Generally, these applications deal with complex data. This paper focuses on covariance matrix estimation problems in non-Gaussian environments and particularly, the $M$-estimators in the context of elliptical distributions. Firstly, this paper extends to the complex case the results of Tyler in \cite{tyler1983robustness}. More precisely, the asymptotic distribution of these estimators as well as the asymptotic distribution of any homogeneous function of degree 0 of the M-estimates are derived.  On the other hand, we show the improvement of such results on two applications: DOA (directions of arrival) estimation using the MUSIC (MUltiple SIgnal Classification) algorithm and adaptive radar detection based on the ANMF (Adaptive Normalized Matched Filter) test. 

\end{abstract}

%\vspace{-0.5cm}
\section{Introduction}
\label{sec:intro}
%\vspace{-0.2cm}
Many signal processing applications require the knowledge of the data covariance matrix. The most often used estimator is the well-known Sample Covariance Matrix (SCM) which is the Maximum Likelihood (ML) estimator for Gaussian data. However, the SCM suffers from major drawbacks. When the data turn out to be non-Gaussian, as for instance in adaptive radar and sonar processing \cite{Kay98}, the performance involved by the SCM can be strongly degraded. Indeed, this is the case in impulsive noise contexts and in the presence of outliers as shown in \cite{Mahot10}. To overcome these problems, there has been an intense research activity in robust estimation theory in the statistical community these last decades \cite{huber2009robust, Hamp86, Maronna06}. Among several solutions, the so-called $M$-estimators originally introduced by Huber \cite{huber1964robust} and investigated in the seminal work of Maronna \cite{Maronna76}, have imposed themselves as an appealing alternative to the classical SCM. They have been introduced within the framework of elliptical distributions.  Elliptical distributions, originally introduced by Kelker in \cite{kelker1970distribution}, encompass a large number of well-known distributions as for instance the Gaussian distribution, or the multivariate Student (or $t$) distribution. They may also be used to model heavy tailed distributions by means of the K-distribution, as may be met for instance in adaptive radar with impulsive clutter \cite{Watts85, conte1991radar, Gini98}. $M$-estimators of the covariance matrix are however seldom used in the signal processing community. Only a limited case, the Tyler's estimator \cite{tyler1987distribution} also called the Fixed Point Estimator \cite{Pascal07} has been widely used as an alternative to the SCM for radar applications. Concerning the $M$-estimators, notable exceptions are the recent papers by Ollila \cite{ollila2012complex, ollila2009influence, ollila2003robust, Ollila2003, ollila2003ieee} who advocates their use in several applications such as array processing. The $M$-estimators have also been studied in the case of large datasets, where the dimension of the data is of the same order as the dimension of the sample \cite{couillet2012robust}.

One possible reason for this lack of interest is that their statistical properties are not well-known in the signal processing community, as opposed to the Wishart distribution of the SCM in the Gaussian context. They have been studied by Tyler \cite{tyler1982radial} in the real case. However, in signal processing applications, data are usually complex and the purpose of this paper is to derive the asymptotic distribution of complex $M$-estimators in the framework of elliptically distributed data. This result is also provided in \cite{ollila2012complex} but without proof. We will also extend to the complex case, a property initially derived by Tyler in \cite{tyler1983robustness}: we show that in the complex elliptical distributions context, the asymptotic distribution of any positive homogeneous functional of degree $0$ of estimates such as $M$-estimates and the SCM, is the same up to a scale factor. This result, useful for applications, extends the one proposed in \cite{ollila2012complex}. 
Thus, for a Gaussian context and for signal processing applications which only need the covariance matrix up to a scale factor, for example Direction-of-Arrival (DOA) estimation or adaptive radar detection, the parameter estimated has the same mean square error when estimated with the SCM or with an $M$-estimator with a few more data (depending on $\sigma_1$). Moreover, when the context is non-Gaussian or contains outliers, the performance obtained with $M$-estimators is scarcely influenced while it is unreliable and possibly completely damaged with the SCM as shown for instance in \cite{Mahot10}.
We illustrate this effect using the MUSIC method and the Adaptive Normalized Matched Filter (ANMF) test introduced by Kraut and Scharf \cite{Kraut01,kraut1999cfar}. It is also illustrated by Ollila in \cite{ollila2009influence}, for MVDR beamforming.\\
This paper is organized as follows. Section \ref{sec:stat} introduces the required background and Section \ref{sec:mest} the known properties of real $M$-estimators. Then Section \ref{sec:main} provides our contribution about the estimators asymptotic distribution. Eventually, in Section \ref{sec:simu}, simulations validate the theoretical analysis and  Section \ref{sec:conclu} concludes this work.\\
Vectors (resp. matrices) are denoted by bold-faced lowercase letters (resp. uppercase letters). $^\ast$, $^T$ and $^H$ respectively represent the conjugate, the transpose and the Hermitian operator. $\sim$ means "distributed as", $\overset{d}{=}$ stands for "shares the same distribution as", $\overset{d}{\rightarrow}$ denotes convergence in distribution and $\otimes$ denotes the Kronecker product. $\vec$ is the operator which transforms a matrix $m\times n$ into a vector of lenth $mn$, concatenating its $n$ columns into a single column. Moreover, $\Ig_m$ is the $m\times m$ identity matrix, $\mathbf{0}_{m,p}$ the $m\times p$ matrix of zeros, $\Jg_{m^2}=\disp\sum_{i}^m\Jg_{ii}\otimes\Jg_{ii}$ where $\Jg_{ii}$ is the $m\times m$ matrix with a one in the $(i,i)$ position and zeros elsewhere and $\Kg$ is the commutation matrix which transforms $\vec(\Ag)$ into $\vec(\Ag^T)$. Eventually, $Im(\yg)$ represents the imaginary part of the complex vector $\yg$ and $Re(\yg)$ its real part.%$\NN(\mug,\Lg)$ denotes the multivariate normal distribution with mean $\mug$ and covariance $\Lg$.

\section{Background}\label{sec:stat}
%\vspace{-0.2cm}

%%%%%%%%%%%%%%%%%%%%%%%%%%%%%%%%%%%%%%%%%%%%%%%%%%%%%%%%%%%%%%%%%%%%%%%%%%%%%%%%%%%%%%%%%%%%%%%%%%%%

\subsection{Elliptical symmetric distribution}
%Let $\xg$ be a $m$-dimensional real random  vector. $\xg$ has an elliptical distribution if its probability density function (PDF) can be written as
%\begin{equation}
%f_{\xg}(\xg)=|\Lg |^{-1/2}g((\xg-\mug)^T\Lg^{-1}(\xg-\mug)),\label{ell-dist}
%\end{equation}
%where $g : [0, \infty)\rightarrow[0,\infty)$ is any function such that (\ref{ell-dist}) defines a PDF, $\mug$ is the statistical mean and $\Lg$ is a scatter matrix. The scatter matrix $\Lg$ reflects the structure of the covariance matrix of $\xg$, i.e. the covariance matrix is equal to $\Lg$ up to a scale factor. This elliptical distribution will be denoted by $E(\mug,\Lg)$.

Let $\zg$ be a $m$-dimensional real (resp. complex circular) random  vector. The vector $\zg$ has a real (resp. complex) elliptical symmetric distribution if its probability density function (PDF) can be written as
\begin{equation}
\begin{array}{l}
g_{\zg}(\zg)=|\Lg |^{-1/2}h_\zg((\zg-\mug)^T\Lg^{-1}(\zg-\mug)),\\\text{ in the real case,}\label{ell-dist-cplx}\\
g_{\zg}(\zg)=|\Lg |^{-1}h_\zg((\zg-\mug)^H\Lg^{-1}(\zg-\mug)),\\\text{ in the complex case,}
\end{array}
\end{equation}
where $h_\zg : [0, \infty)\rightarrow[0,\infty)$ is any function such that (\ref{ell-dist-cplx}) defines a PDF, $\mug$ is the statistical mean and $\Lg$ is a scatter matrix. The scatter matrix $\Lg$ reflects the structure of the covariance matrix of $\zg$, i.e. the covariance matrix is equal to $\Lg$ up to a scale factor. This real (resp. complex) elliptically symmetric distribution will be denoted by $\E(\mug,\Lg,h_{\zg})$ (resp. $\CE(\mug,\Lg,h_{\zg})$). One can notice that the Gaussian distribution is a particular case of elliptical distributions. A survey on complex elliptical distributions can be found in \cite{ollila2012complex}.

In this paper, we will assume that $\mug=\mathbf{0}_{m,1}$. Without loss of generality, the scatter matrix will be taken to be equal to the covariance matrix when the latter exists. Indeed, when the second moment of the distribution is finite, function $h_\zg$ in (\ref{ell-dist-cplx}) can always be defined such that this equality holds. If the distribution of the data has a none finite second-order moment, then we will simply consider the scatter matrix estimator.

%%%%%%%%%%%%%%%%%%%%%%%%%%%%%%%%%%%%%%%%%%%%%%%%%%%%%%%%%%%%%%%%%%%%%%%%%%%%%%%%%%%%%%%%%%%%%%%%%%%%
\subsection{Generalized Complex Normal distribution}
As written before, the Gaussian distribution is a particular case of elliptical symmetric distributions. However, in the complex framework, it is true only for circular Gaussian random vectors. We now present the generalization of this distribution as presented by Van den Bos in \cite{van1995multivariate}.

Let $\zg=\xg+j\yg$ be a $m$-dimensional complex random  vector. The vector $\zg$ is said to have a generalized complex normal distribution if and only if $\vg=(\xg^T,\yg^T)^T \in \RR^{2m}$ has a normal distribution. This generalized complex normal distribution will be denoted by $\GCN(\mug,\Sig,\Omg)$ where $\mug$ is the mean, $\Sig=E[(\zg-\mug)(\zg-\mug)^H]$ the covariance matrix, and $\Omg=E[(\zg-\mug)(\zg-\mug)^T]$ the pseudo-covariance matrix.

\subsection{$M$-estimators of the scatter matrix}
%\vspace{-0.1cm}
\label{sect:Mestim}

Let $(\zg_1, ... ,\zg_N)$ be an $N$-sample of $m$-dimensional real (resp. complex circular) independent vectors with $\zg_i \sim \E(\mathbf{0}_{m,1}, \Lg,h_{\zg})$ (resp. $\zg_i \sim \CE(\mathbf{0}_{m,1}, \Lg,h_{\zg})$), $i=1,...,N$. The real (resp. complex) $M$-estimator of $\Lg$ is defined as the solution of the following equation
\begin{equation}\label{eq:VN}
\Mc=\frac{1}{N}\sum_{n=1}^N u\be\zg_n'\Mc^{-1}\zg_n\en\zg_n\zg_n'.
\end{equation}
where the symbol $'$ stands for $^T$ in the real case and for $^H$ in the complex one.

$M$-estimators have first been studied in the real case, defined as solution of (\ref{eq:VN}) with real samples. Existence and uniqueness of the solution of (\ref{eq:VN}) has been shown in the real case, provided function $u$ satisfies a set of general assumptions stated by Maronna in \cite{Maronna76}. These conditions have been extended to the complex case by Ollila in \cite{ollila2003robust}. They are recalled here below in the case where $\mug=\mathbf{0}_{m,1}$:
%\vspace{-3mm}
\begin{itemize} \item[-] $u$ is non-negative, non increasing, and continuous on $[0,\infty)$.
\item[-] Let $\psi(s)=s\, u(s)$ and $K=\sup_{s\geq0}\psi(s)$.  $m<K<\infty$, $\psi$ is increasing, and strictly increasing on the interval where $\psi<K$.
\item[-] Let $P_N(.)$ denote the empirical distribution of $(\zg_1, ... ,\zg_N)$. There exists $a > 0$ such that for every hyperplane $S$, $\dim(S) \leq m-1$, $P_N(S)\leq 1-\frac{m}{K}-a$.
This assumption can be strongly relaxed as shown in \cite{tyler1988some, Kent91}. 
\end{itemize}
Let us now consider the following equation, which is roughly speaking the limit of (\ref{eq:VN}) when $N$ tends to infinity:
\begin{equation}\label{eq:V}
\Mg=E\left[u(\zg'\Mg^{-1}\zg)\,\zg\zg'\right],
\end{equation}
where $\zg\sim \E(\mathbf{0}_{m,1},\Lg,h_{\zg})$ (resp. $\CE(\mathbf{0}_{m,1}, \Lg, h_{\zg})$) and where the symbol $'$ stands for $^T$ in the real case and for $^H$ in the complex one.

Then, under the above conditions, it has been shown for the real case in \cite{Kent91, Maronna76} that:
%\vspace{-2mm}
\begin{itemize} \item[-] Equation (\ref{eq:V}) (resp. (\ref{eq:VN})) admits a unique solution $\Mg$ (resp. $\Mc$) and
\begin{equation}\Mg=\sigma^{-1}\boldsymbol{\Lambda}, \label{eq:VsL}\end{equation}
where $\sigma$ is  the solution of $E[\psi(\sigma |\tg|^2)]=m$, where $\tg \sim \E (\mathbf{0}_{m,1},\Ig_m,h_{\zg})$, see e.g. \cite{Maronna06}  (resp. $\tg \sim\CE(\mathbf{0}_{m,1}, \Ig_m, h_{\zg})$).
\item[-] A simple iterative procedure provides $\Mc$.
\item[-] $\Mc$ is a consistent estimate of $\Mg$. 
\end{itemize}
The extension to the complex case of previous results has been done in \cite{ollila2012complex}.
%%%%%%%%%%%%%%%%%%%%%%%%%%%%%%%%%%%%%%%%%%%%%%%%%%%%%%%%%%%%%%%%%%%%%%%%%%%%%%%%%%%%%%%%%%%%%%%%%%%%
\subsection{Wishart distribution}
	
The real (resp.complex) Wishart distribution $\WW(N,\boldsymbol{\Lambda})$ (resp. $\CW(N,\boldsymbol{\Lambda})$) is the distribution of $\disp\sum_{n=1}^N \zg_n\zg_n'$, where $\zg_n$ are real (resp. complex circular), independent identically distributed (i.i.d), Gaussian with zero mean and covariance matrix $\Lg$. Let $ \disp\Wg_N = N^{-1} \sum_{n=1}^N \zg_n
\zg_n' $ be the related SCM which will be also referred to, as a Wishart matrix. The asymptotic distribution of the Wishart matrix $\Wg_N$ is (see e.g. \cite{bilodeau1999theory})
\begin{equation}
\begin{array}{l}
\sqrt{N}\vec(\Wg_N-\Lg)\overset{d}{\longrightarrow} \NN \be \mathbf{0}_{m^2,1},(\Lg\otimes\Lg)(\Ig_{m^2}+\Kg)\en \\\text{ in the real case,}\\
\noindent\sqrt{N}\vec(\Wg_N-\Lg)\overset{d}{\longrightarrow}\GCN \be \mathbf{0}_{m^2,1},\Lg^T\otimes\Lg,(\Lg^T\otimes\Lg)\Kg\en\\\text{ in the complex case.}
\label{eq:Westim}
\end{array}
\end{equation}

We now introduce real $M$-estimators asymptotic properties since they are used as a basis for the extension to the complex case.
%%%%%%%%%%%%%%%%%%%%%%%%%%%%%%%%%%%%%%%%%%%%%%%%%%%%%%%%%%%%%%%%%%%%%%%%%%%%%%%%%%%%%%%%%%%%%%%%%%%%
\section{Real $M$-estimators properties}
\label{sec:mest}

\subsection{Asymptotic distribution of the real $M$-estimators}
Let $\Mc$ be a real $M$-estimator following Maronnas's conditions \cite{Maronna76}, recalled in section \ref{sect:Mestim}. The asymptotic distribution of $\Mc$ is given by Tyler in \cite{tyler1982radial}:
\begin{equation}
\sqrt{N}\vec(\Mc-\Mg)\overset{d}{\longrightarrow}\NN \be\mathbf{0}_{m,1},\Pig\en,
\label{eq:var_reel}
\end{equation}
where $\Pig=\sigma_1(\Ig_{m^2}+\Kg)(\Mg \otimes \Mg)+\sigma_2\vec(\Mg)\vec(\Mg)^T$, $\sigma_{1}$ and $\sigma_{2}$ are given by (\cite{tyler1982radial}):

\begin{equation}
\label{eq:sig12}
\ba
\begin{array}{l}
\sigma_1=a_1(m+2)^2(2a_2+m)^{-2},\\
\sigma_2=a_2^{-2}\left[(a_1-1)-\cfrac{2a_1(a_2-1)}{(2a_2+m)^{2}}\left[m+(m+4)a_2\right]\right],
\end{array}
\right.
\end{equation}
with
\begin{equation*}
\ba
\begin{array}{l}
a_1=[m(m+2)]^{-1}\,E\left[ \psi^2(\sigma |\tg|^2)\right],\\
a_2=m^{-1}\,E[\sigma |\tg|^2\psi'(\sigma |\tg|^2)],\\
\end{array}
\right.
\end{equation*}
and $\sigma$ is given in equation (\ref{eq:VsL}).

%%%%%%%%%%%%%%%%%%%%%%%%%%%%%%%%%%%%%%%%%%%%%%%%%%%%%%%%%%%%%%%%%%%%%%%%%%%%%%%%%%%%%%%%%%%%%%%%%%%

\subsection{An important property of real $M$-estimators}
\label{subsect:reel}
Let $\Vg$  be a  fixed symmetric positive-definite matrix and $\Vg_N$ a sequence of symmetric positive definite random matrices of order $m$ which satisfies %the following conditions.
%\paragraph*{Condition 1}  
\begin{equation}
\sqrt{N}\vec(\Vg_N-\Vg)\overset{d}{\longrightarrow}\NN \be\mathbf{0}_{m^2,1},\Sg\en,
\label{eq:C1}
\end{equation}
where $\Sg=\mu_1(\Ig_{m^2}+\Kg)(\Vg \otimes \Vg)+\mu_2\vec(\Vg)\vec(\Vg)^T$, $\mu_1$ and $\mu_2$ are any real numbers such that $\Sg$ is a positive matrix.
%\paragraph*{Condition 2}
%For any $\Ag$ such that $\Ag\Ag^T=\Mg^{-1}$, the distribution of $\Xg=\Ag\NN\be\mathbf{0},\Pig\en\Ag^T$ is invariant under the transformation $\Xg\longrightarrow\Qg\Xg\Qg^T$ for any orthogonal $\Qg$.

%One can notice that in a Gaussian context, the SCM satisfies these conditions: from (\ref{eq:Westim}), $\mu_1=1$ and $\mu_2=0$. This is also the case for real $M$-estimators: from (\ref{eq:var_reel}), $\mu_1=\sigma_1$ and $\mu_2=\sigma_2$.

Let $H(\Vg)$ be a $r$-dimensional multivariate function on the set of $m\times m$ positive-definite symmetric matrices with continuous first partial derivatives and such as $H(\Vg)=H(\alpha \Vg)$ for all $\alpha>0$. Then under conditions (\ref{eq:C1}), Tyler has shown in \cite{tyler1983robustness} Theorem 1, that
\begin{multline}
\sqrt{N}\be H(\Vg_N)-H(\Vg) \en \overset{d}{\longrightarrow}\\ \NN \be \mathbf{0}_{r,1},2 \mu_1 H'(\Vg)(\Vg\otimes\Vg)H'(\Vg)^T\en,
\end{multline}
where $H'(\Vg)= \cfrac{1}{2}\left(\cfrac{dH(\Vg)}{d \vec(\Vg)}\right)\left(\Ig_{m^2}+\Jg_{m^2}\right)$.

By noticing that, in a Gaussian context the SCM satisfies $\mu_1=1$ and $\mu_2=0$ (equation (\ref{eq:Westim})) and that real $M$-estimators verify $\mu_1=\sigma_1$ and $\mu_2=\sigma_2$ (equation (\ref{eq:var_reel})), Tyler's theorem shows that $\sqrt{N}(H(\Wg_N)-H(\Lg))$ and $\sqrt{N/\sigma_1}(H(\Mc)-H(\Lg))$ share the same asymptotic distribution.\\

In practice, $H(.)$ may be a function which associates a parameter of interest to a covariance matrix. This scale-invariant property has also been exploited in \cite{ollila2009influnece}. The concerned signal processing applications are those in which multiplying the covariance matrix by a positive scalar does not change the result. This is the case for instance for the MUSIC method in which the estimated parameters are the signals DOA. An other example is given by adaptive radar processing in which the parameter is the ANMF test statistic \cite{Kraut01,kraut1999cfar}. Here, $H$ is defined by: $\Mc\overset{H}{\rightarrow}H(\Mc)=\cfrac{|\pg^H\Mc^{-1}\yg|^2}{(\pg^H\Mc^{-1}\yg)(\yg^H\Mc^{-1}\yg)}$. \\

The aim of the next section is to extend those results to the complex case, which is the frequently met framework for most signal processing applications.

%%%%%%%%%%%%%%%%%%%%%%%%%%%%%%%%%%%%%%%%%%%%%%%%%%%%%%%%%%%%%%%%%%%%%%%%%%%%%%%%%%%%%%%%%%%%%%%%%%%ù
\section{Main results in complex case}
 \label{sec:main}
\subsection{Asymptotic distribution of the complex $M$-estimator}
Let $(\zg_1, ... ,\zg_N)$ be an $N$-sample of $m$-dimensional complex independent vectors with $\zg_n \sim \CE(\mathbf{0}_{m,1}, \Lg,h_{\zg})$, $n=1,...,N$. We consider the complex $M$-estimator $\Mc_C$ which verifies equation (\ref{eq:VN}), and we denote $\Mg_C$ the solution of (\ref{eq:V}).

%%%%%%%%%%%%%%%%%%%%Théoreme%%%%%%%%%%%%%%%%%%%%%%%%%%%%%%%%%%%%%%%%%%%%%%
\begin{Thm}
\label{sect:theorem}
The asymptotic distribution of $\Mc_C$ is given by
\begin{equation}
\sqrt{N}\vec(\Mc_C-\Mg_C)\overset{d}{\longrightarrow}\GCN \be\mathbf{0}_{m^2,1},\Sig,\Omg\en,
\label{eq:Mestim}
\end{equation}
where $\Sig$ and $\Omg$ are defined by
\begin{equation}
\begin{array}{ll}
&\Sig=\sigma_1 \Mg_C^{T}\otimes\Mg_C+\sigma_2\vec(\Mg_C)\vec(\Mg_C)^H,\\
&\Omg=\sigma_1 (\Mg_C^{T}\otimes\Mg_C)\Kg+\sigma_2\vec(\Mg_C)\vec(\Mg_C)^T,
\end{array}\label{eq:varcplx}
\end{equation}
with
\begin{equation}
\label{eq:sig12}
\left\lbrace
\begin{array}{l}
\sigma_1=a_1(m+1)^2(a_2+m)^{-2},\\
\sigma_2=a_2^{-2}\left[(a_1-1)-\cfrac{2a_1(a_2-1)}{(2a_2+2m)^{2}}\left[2m+(2m+4)a_2\right]\right],
\end{array}\right.
\end{equation}
and
\begin{equation*}
\left\lbrace
\begin{array}{l}
a_1=[m(m+1)]^{-1}\,E\left[ \psi^2(\sigma |\tg|^2)\right],\\
a_2=m^{-1}\,E[\sigma |\tg|^2\psi'(\sigma |\tg|^2)],\\
\end{array}\right.
\end{equation*}
where $\sigma$ is  the solution of $E[\psi(\sigma |\tg|^2)]=m$, where $\tg \sim \CE (\mathbf{0}_{m,1},\Ig_m,h_{\zg})$.
\end{Thm}

This result is also given in \cite{ollila2012complex} with others assumptions but without proof.

%%%%%%%%%%%%%%%%%%%%%%%%%%%%%%%%%%%%%Preuve Théoreme1%%%%%%%%%%%%%%%%%%%%%%%%%
\subsection{Proof of Theorem \ref{sect:theorem}}
\subsubsection{Notations}
Let us first introduce the following linear one-to-one transformation of a Hermitian $m\times m$ matrix $\Ag$ into a real symmetric $2m\times 2m$ matrix:
\begin{equation}
f(\Ag)=\cfrac{1}{2}\begin{pmatrix}
Re(\Ag) & -Im(\Ag)\\
Im(\Ag) & Re(\Ag)
\end{pmatrix}
.
\end{equation}
The inverse transformation is given by $\Ag=\gg^H f(\Ag) \gg$ where $\gg^T= (\Ig_m\,,\, -j\Ig_m)$. Function $f$ has some useful properties. Let $\ug_n$ and $\vg_n$ be the following $2m$ vectors:
\begin{equation}
\begin{array}{l}
\ug_n=(Re(\zg_n)^T,Im(\zg_n)^T)^T\\
\vg_n=(-Im(\zg_n)^T,Re(\zg_n)^T)^T,
\end{array}
\end{equation}
which are both distributed according to $\E(\mathbf{0}_{2m,1}, \Lg_R,g_{\zg})$ where $\Lg_R= f(\Lg)$.

Then, it may be shown that
\begin{equation}
\nonumber
\begin{array}{l}
f(\Ag^{-1})=\cfrac{1}{4}\,f(\Ag)^{-1}\\
f(\zg_n\zg_n^H)=\cfrac{1}{2} \, (\ug_n\ug_n^T+\vg_n\vg_n^T)\\
\zg_n^H\Ag^{-1}\zg_n=\cfrac{1}{2}\,\ug_n^T f(\Ag)^{-1}\ug_n=\cfrac{1}{2}\,\vg_n^T f(\Ag)^{-1}\vg_n
\end{array}
\end{equation}
Let $\Trg = \begin{pmatrix}
\mathbf{0}_{m,m} & -\Ig_m \\
\Ig_m& \mathbf{0}_{m,m}
\end{pmatrix}
$, one has $\vg_n=\Trg\ug_n$ and $\ug_n=\Trg^T\vg_n.$

Let us also introduce
\begin{equation}
\begin{array}{l}
\Mc_R=f(\Mc_C),\\
\Mg_R=f(\Mg_C).
\end{array}
\label{eq:MR}
\end{equation}

It is easy to show that equation (\ref{eq:VN}) defining the complex $M$-estimator $\Mc_C$, is equivalent to the following equation involving $\Mc_R$:
\begin{multline}
\Mc_R=\cfrac{1}{2N}\disp\sum_{n=1}^N\be u_r\be\ug_n^T\Mc_R^{-1}\ug_n\en\ug_n\ug_n^T\right.+\\ \left. u_r\be\vg_n^T\Mc_R^{-1}\vg_n\en\vg_n\vg_n^T\en,
\label{eq:MR2}
\end{multline}
where $u_r(s)=u(s/2)$. Roughly speaking, equation (\ref{eq:MR2}) defines a real $M$-estimator involving the $2N$ real samples $\ug_n$ and $\vg_n$.

Let $\Mc_u$ and $\Mc_v$ be respectively the two $M$-estimators defined by
\begin{equation}
\begin{array}{l}
\Mc_u=\cfrac{1}{N}\disp\sum_{n=1}^N u_r\be\ug_n^T\Mc_u^{-1}\ug_n\en\ug_n\ug_n^T,\\
\Mc_v=\cfrac{1}{N}\disp\sum_{n=1}^N u_r\be\vg_n^T\Mc_v^{-1}\vg_n\en\vg_n\vg_n^T,
\label{eq:Mcu}
\end{array}
\end{equation}
and let $\Mg_u$, $\Mg_v$ be the associated solutions of
\begin{equation}
\begin{array}{l}
\Mg_u=E\left[ u_r\be\ug_1^T\Mg_u^{-1}\ug_1\en\ug_1\ug_1^T\right],\\
\Mg_v=E\left[ u_r\be\vg_1^T\Mg_v^{-1}\vg_1\en\vg_1\vg_1^T\right].
\end{array}
\nonumber
\end{equation}
By applying $\Trg$ on equation (\ref{eq:Mcu}), one obtains 
\begin{equation}
\Mc_v=\Trg\Mc_u\Trg^T.
\label{eq:Mv=TMuT}
\end{equation}
Moreover, since $\vg_n$ has the same distribution as $\ug_n$, 
\begin{equation}
\Mg_u=\Mg_v =\Trg\Mg_u\Trg^T.
\label{eq:Mv=Mu}
\end{equation}

\subsubsection{An intermediate result}
\begin{Lem}\label{lem1}
% The following lemma is required for proving Theorem \ref{sec:theorem}.
$\Mc_R$ and $\cfrac{1}{2}(\Mc_u+\Mc_v)$ have the same Gaussian asymptotic distribution.
\end{Lem}

\begin{proof}
See appendix \ref{appendix}.
\end{proof}

\subsubsection{End of proof of theorem \ref{sect:theorem}}

By using equation (\ref{eq:MR}) and the inverse of $f$, one obtains $\Mc_C=\gg^H\Mc_R\gg$. From the Lemma \ref{lem1} $\vec(\Mc_R)$. has a normal distribution. It follows that $\vec(\Mc_C)$ has a generalized complex normal distribution.

Given the property $\vec(\Ag\Bg\Cg)=(\Cg^T\otimes\Ag)\vec(\Bg)$ where $\Ag$, $\Bg$, $\Cg$ are 3 matrices, and using the fact that $\Mc_C=\gg^H\Mc_R\gg$, one has 
\begin{equation}
\begin{array}{lll}
\Sig &=& N E[\vec(\Mc_C-\Mg_C)\vec(\Mc_C-\Mg_C)^H]\\
&=&(\gg^T\otimes\gg^H)\, E[N\vec(\Mc_R-\Mg_R)\vec(\Mc_R-\Mg_R)^H]\\&&(\gg^T\otimes\gg^H)^H.
\label{eq:Sig1}
\end{array}
\end{equation}

Using lemma \ref{lem1}, and the equalities (\ref{eq:Mv=TMuT}) and (\ref{eq:Mv=Mu}), equation (\ref{eq:Sig1}) gives
\begin{equation}
\begin{array}{lll}
\Sig &= &(\gg^T\otimes\gg^H)N E\left[\vec\be\cfrac{1}{2}(\Mc_u+\Trg\Mc_u\Trg^T)-\Mg_u\en \right. \\ &&\left.\vec\be\cfrac{1}{2}(\Mc_u+\Trg\Mc_u\Trg^T)-\Mg_u\en^H\right](\gg^T\otimes\gg^H)^H\\
&=& (\gg^T\otimes\gg^H)(\Ig_{4m^2}+\Trg\otimes\Trg)\Pig_u (\Ig_{4m^2}+\Trg\otimes\Trg)^H\\&&(\gg^T\otimes\gg^H)^H\\
&=&(\gg^T\otimes\gg^H)\, \Pig_u \,(\gg^T\otimes\gg^H)^H,
\label{eq:Sig2}
\end{array}
\end{equation}
by using $\Pig_u$ the asymptotic covariance of $\Mg_u$ and the equalities $\gg^T\Trg=-j\gg^T$ and $\gg^H\Trg=j\gg^H$.

Using the expression given in (\ref{eq:var_reel}), and taking into account that the $\ug_n$ are $2m$-dimensional vectors, we have 
\begin{equation}
\Pig_u=\sigma_1(\Ig_{4m^2}+\Kg)(\Mg_u \otimes \Mg_u)+\sigma_2\vec(\Mg_u)\vec(\Mg_u)^T
 \end{equation}
 where $\sigma_1$ and $\sigma_2$ will be specified later.
  
A consequence of lemma \ref{lem1} is that $\Mg_R=\Mg_u$. Indeed, from the definition of $\Mc_R$, one has 
\begin{equation*}
\Mg_R = \cfrac{1}{2} \be E\left[u_r\be\ug^T\Mg_R^{-1}\ug\en\ug\ug^T\right]+ E\left[u_r\be\vg^T\Mg_R^{-1}\vg\en\vg\vg^T\right]\en
\end{equation*}
The first term of the right hand side is the definition of $\Mg_u$ while the second one is the one of $\Mg_v$. Then, as $\Mg_u=\Mg_v$, one has $\Mg_R=\Mg_u$.

Therefore, $\Mg_C=\gg^H\Mg_u\gg$, which leads to
\begin{equation}
\Sig=\sigma_1(\Mg_C^T \otimes \Mg_C)+\sigma_2\vec(\Mg_C)\vec(\Mg_C)^H.
 \end{equation}
 
 Now let us turn to the $\sigma_1$ and $\sigma_2$ coefficients. Using (\ref{eq:var_reel}), one has
\begin{equation}
\left\lbrace
\begin{array}{l}
\sigma_1=a_1(2m+2)^2(2a_2+2m)^{-2},\\
\sigma_2=a_2^{-2}\left[(a_1-1)-\cfrac{2a_1(a_2-1)}{(2a_2+2m)^{2}}\left[2m+(2m+4)a_2\right]\right],\\
\\
a_1=[2m(2m+2)]^{-1}\,E\left[ \psi_r^2(\sigma |\sg|^2)\right],\\
a_2=(2m)^{-1}\,E[\sigma |\sg|^2\psi_r'(\sigma |\sg|^2)],\\
\end{array}\right.
\label{sig12preuve}
\end{equation}
where $\psi_r(s)=s u_r(s)$, $\sg \sim \E (\mathbf{0}_{2m,1},\Ig_{2m},h_{\zg})$ and $\sigma$ is the solution of 
\begin{equation}
E[\psi_r(\sigma |\sg|^2)]=2m.
\label{eq:sigma_r}
\end{equation}

Since $\psi_r(s)=2\psi(s/2)$, equation (\ref{eq:sigma_r}) is equivalent to
\begin{equation}
E \left[\psi\be\cfrac{\sigma}{2} |\sg|^2\en\right]=m.
\label{eq:sigma_r2}
\end{equation}
Moreover let $\tg \sim \CE (\mathbf{0}_{m,1},\Ig_{m},h_{\zg})$. Then $|\tg|^2$ has the same distribution as $|\sg|^2/2$ so that (\ref{eq:sigma_r}) and (\ref{eq:sigma_r2}) are also equivalent to
\begin{equation}
E\left[\psi(\sigma |\tg|^2)\right]=m
\end{equation}
We finally obtain the expression of $\Sig$.

\paragraph{Asymptotic pseudo-covariance matrix}
$\Omg$ is defined as
\begin{equation}
\Omg = N E[\vec(\Mc_C-\Mg_C)\vec(\Mc_C-\Mg_C)^T]
\end{equation}
Using the commutation matrix $\Kg$, one has
\begin{equation}
\begin{array}{ll}
\Kg\vec(\Mc_C-\Mg_C)&=\vec(\Mc_C^T-\Mg_C^T)\\
&=\vec(\Mc_C-\Mg_C)^\ast
\end{array}
\end{equation}
since $\Mc_C$ is Hermitian. Thus one can write
\begin{equation}
\begin{array}{ll}
\vec(\Mc_C-\Mg_C)^T &=[\vec(\Mc_C-\Mg_C)^\ast]^H,\\
&=\vec(\Mc_C-\Mg_C)^H \Kg,
\end{array}
\end{equation}
where $\Kg^H=\Kg$.

Therefore, $\Omg=\Sig \Kg$, which leads to the result of theorem \ref{sect:theorem} after a few derivations, and concludes the proof.

%%%%%%%%%%%%%%%%%%%%%%%%%%%%%%%%%%%%%%%%%%%%%%%%%%%%%%%%%%%%%%%%%%%%%%%
In the following part, we extend the result of section \ref{subsect:reel} to the complex case.

\subsection{An important property of complex $M$-estimators}
\begin{Thm}\hspace{1cm}
\label{sect:theorem2}

\begin{itemize}
\item Let $\Mg$  be a  fixed Hermitian positive-definite matrix and $\Mc$  a sequence of Hermitian positive-definite random matrix estimates of order $m$ which satisfies% \emph{condition 2} and
%\paragraph*{condition 3}
\begin{equation}
\sqrt{N}\be\vec(\Mc-\Mg)\en\overset{d}{\longrightarrow}\GCN \be\mathbf{0}_{m^2,1},\Sig_M, \Omg_M \en,
\end{equation}
with
\begin{equation}
\begin{array}{l}
\Sig_M=\nu_1 \Mg^{T}\otimes\Mg+\nu_2\vec(\Mg)\vec(\Mg)^H,\\
\Omg_M=\nu_1 (\Mg^{T}\otimes\Mg) \Kg+\nu_2\vec(\Mg)\vec(\Mg)^T,
\end{array}
\end{equation}
where $\nu_1$ and $\nu_2$ are any real numbers.\\

\item Let $H(\Mg)=(h_1,..., h_r)^T$ be a $r$-dimensional multivariate function on the set of $m\times m$ complex Hermitian positive-definite matrices, possessing continuous first partial derivatives and such as $H(\Mg)=H(\alpha \Mg)$ for all $\alpha>0$.
\end{itemize}
Then,
\begin{equation}
\sqrt{N}\be H(\Mc)-H(\Mg) \en \overset{d}{\longrightarrow} \GCN \be \mathbf{0}_{r,1},  \Sig_H, \Omg_H\en,
\label{eq:th2}
\end{equation}
where $\Sig_H$ and $\Omg_H$ are defined as
\begin{equation}
\begin{array}{l}
\Sig_H= \nu_1 H'(\Mg)(\Mg^T\otimes\Mg)H'(\Mg)^H,\\
\Omg_H=  \nu_1 H'(\Mg)(\Mg^T\otimes\Mg)\Kg H'(\Mg)^T,
\end{array}
\end{equation}
and $H'(\Mg)=\cfrac{dH(\Mg)}{d \vec(\Mg)}=(h'_{ij})$ with $h'_{ij}=\cfrac{\partial h_i}{\partial m_j}$ where $\vec(\Mg)=(m_i)$.
\end{Thm}

%where $H'(\Mg)= \cfrac{1}{2}\left(\cfrac{dH(\Mg)}{d \vec\Mg}\right)\left(\Ig+\Jg_m\right)$, with $\Jg_m=\disp\sum_{i}\Jg_{ii}\otimes\Jg_{ii}$ and $\Jg_{ii}$ is the $m\times m$ matrix with a one in the $(i,i)$ position and zeros elsewhere.

%Moreover, normalized $M$-estimators have the same asymptotic distribution than a real normalized Wishart matrix, up to the scale factor $\sigma_1$.\\
\begin{proof}
One can first notice that $H'(\Mg)\vec(\Mg)=\mathbf{0}_{r,1}$. Indeed, since $H(\Mg)=H(\alpha\Mg)$ for all $\alpha>0$, the subspace generated by the vector $\vec(\Mg)$ is an iso-$H$ region. Therefore, $H'(\Mg)$ which can be seen as a gradient of $H$, is orthogonal to $\vec(\Mg)$.

A first order approximation of $H(\Mc)$ gives
\begin{equation}
H(\Mc)\simeq  H(\Mg)+ H'(\Mg)\vec(\Mc-\Mg),
\end{equation}

Thus one has,
\begin{equation}
\begin{array}{ll}

\Sig_H &= N E\left[\be H(\Mc)-H(\Mg)\en \be H(\Mc)-H(\Mg)\en ^H\right]\\
&=H'(\Mg)E\left[N \vec(\Mc-\Mg)\vec(\Mc-\Mg)^H\right]H'(\Mg)^H\\
&=H'(\Mg)\Sig_M H'(\Mg)^H\\
&=H'(\Mg)\be\nu_1 \Mg^{T}\otimes\Mg+\nu_2\vec(\Mg)\vec(\Mg)^H\en H'(\Mg)^H\\
&=H'(\Mg)\be\nu_1 \Mg^{T}\otimes\Mg\en H'(\Mg)^H.
\end{array}
\end{equation}
The proof is similar for $\Omg_H$.
\end{proof}

Similarly to the real case, when the data have a complex Gaussian distribution, the SCM is a complex Wishart matrix. Moreover, the SCM estimator verifies the conditions of the theorem and its coefficients $(\mu_1,\mu_2)$ are equal to $(1,0)$. Complex normalized $M$-estimators also verify the conditions of the theorem with $(\mu_1,\mu_2)=(\sigma_1,\sigma_2)$. Thus they have the same asymptotic distribution as the complex normalized Wishart matrix, up to a scale factor $\sigma_1$ depending on the considered $M$-estimator. The same conclusion holds for the Fixed Point Estimator \cite{tyler1987distribution, Pascal-07} since it verifies the assumptions of theorem \ref{sect:theorem2} (see \cite{Pascal07b} for its asymptotic distribution).\\

\section{Simulations}
\label{sec:simu}
The results of this paper are illustrated using  the complex analogue of Huber's $M$-estimator as described in \cite{ollila2003robust}. The corresponding weight function $u(.)$ of equation (\ref{eq:VN}) is defined by
\begin{equation}
u(s)=\cfrac{1}{\beta} \, \mbox{min}(1,k^2/s),
\end{equation}
where $k^2$ and $\beta$ depend on a single parameter $0<q< 1$, according to
\begin{equation}q=F_{2m}(2k^2),\label{eq:q}\end{equation}
\begin{equation}\beta=F_{2m+2}(2k^2)+k^2\cfrac{1-q}{m}\label{eq:beta}\end{equation}
where $F_m(.)$ is the cumulative distribution function of a $\chi^2$ distribution with $m$ degrees of freedom. Thus Huber estimate is the solution of
\begin{multline}
\label{eq:hub}
\Mc_{Hub}=\cfrac{1}{N \beta}\dsum_{i=1}^{N}\left[\zg_i\zg_i^H\mathbbm{1}_{\zg_i^H\Mc_{Hub}^{-1}\zg_i\leq k^2}\right.\\\left.+k^2\cfrac{\zg_i\zg_i^H}{\zg_i^H\Mc_{Hub}^{-1}\zg_i}\mathbbm{1}_{\zg_i^H\Mc_{Hub}^{-1}\zg_i> k^2} \right],
\end{multline}
which can be rewritten
\begin{multline}
\label{eq:hub}
\Mc_{Hub}=\cfrac{1}{N \beta}\dsum_{i=1}^{N}\left[\zg_i\zg_i^H\mathbbm{1}_{\zg_i^H\Mc_{Hub}^{-1}\zg_i\leq k^2}\right]\\+\cfrac{1}{N \beta}k^2\dsum_{i=1}^{N}\left[\cfrac{\zg_i\zg_i^H}{\zg_i^H\Mc_{Hub}^{-1}\zg_i}\mathbbm{1}_{\zg_i^H\Mc_{Hub}^{-1}\zg_i> k^2} \right],
\end{multline}
where $\mathbbm{1}$ is the indicator function.

The first summation corresponds to unweighted data which are treated as in the SCM; the second one is associated to normalized data treated as outliers. In a complex Gaussian context and when $N$ tends to infinity, it may be shown that the proportion of data treated with the SCM is equal to $q$. Moreover the choice of $k^2$ and $\beta$ according to (\ref{eq:q}) and (\ref{eq:beta}), leads to a consistant $M$-estimator of the covariance matrix ($\sigma=1$ in equation (\ref{eq:VsL})).

In the following simulations, $q=0.75$.

\subsection{Asymptotic performance of DOA estimated by the MUSIC method, with the SCM and Huber's $M$-estimator}

Now let us turn to theorem \ref{sect:theorem2}. To illustrate our result, we consider a simulation using the MUltiple SIgnal Classification (MUSIC) method, which estimates the Directions Of Arrival (DOAs) of a signal. We consider in this paper a single signal to detect. However, the multi-sources case can be similarly analyzed. Under this assumption, let us define $H(\Mc)$ the estimated DoAs obtained from the MUSIC pseudo-spectrum: $H(\Mc)=\thetac$. 

A $m=3$ uniform linear array (ULA) with half wavelength sensors spacing is used, which receives a Gaussian stationnary narrowband signal with DOA $20^\circ$. The array output is corrupted by an additive noise which is firstly spatially white Gaussian and secondly K-distributed with shape parameter 0.1. Moreover, the SNR per sensor is 5dB and the $N$ snapshots are assumed to be independent. 
The MUSIC method uses the estimation of the covariance matrix with the $N$ snapshots and here, the employed covariance matrix estimators are the SCM and the complex analogue of Huber's $M$-estimator as defined in equation (\ref{eq:hub}).

Figure \ref{fig3} depicts the Root Mean Square Error (RMSE) in degrees, of the DOA estimated with $N$ data for the SCM and for Huber's estimate, when the additive noise is white Gaussian. The RMSE of the DOA estimated with $\sigma_1 \, N$ data with Huber's estimate is also represented. We observe that for $N$ large enough ($N\geq40$), this curve and the SCM one overlap, as expected from theorem (\ref{sect:theorem}). In this example, $\sigma_1=1.067$.
%\vspace{-0.5cm}

\begin{figure}[H]
\begin{center}
\includegraphics[width=1\linewidth]{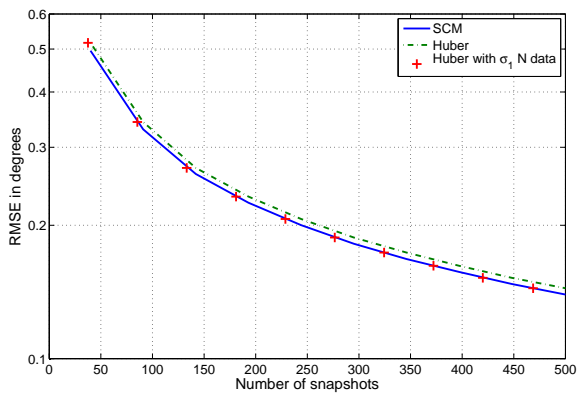}
\caption{One source DOA RMSE ($m=3$ antennas) for Huber's estimate and the SCM, for spatially white Gaussian additive noise.}
\label{fig3}
\end{center}
\end{figure}

Figure \ref{fig4} depicts the RMSE of the DOA estimated with $N$ data for the SCM and for Huber's estimate, when the additive noise is K-distributed with shape parameter 0.1. A shape parameter close to 1  ($\gtrsim 0.9$) indicates a distribution close to the Gaussian distribution whereas it indicates an impulsive noise when the parameter is close to 0 ($\lesssim 0.1$). Thus, the noise being quite impulsive in our example,  we observe that the RMSE of Huber's $M$-estimator is smaller than the SCM, the latter giving worse results than in the Gaussian case. It points out the fact that the SCM gives poor results as soon as the context is far from a Gaussian environment whereas Huber's $M$-estimator is more robust and much more interesting in that case.
%\vspace{-0.5cm}

\begin{figure}[H]
\begin{center}
\includegraphics[width=1\linewidth]{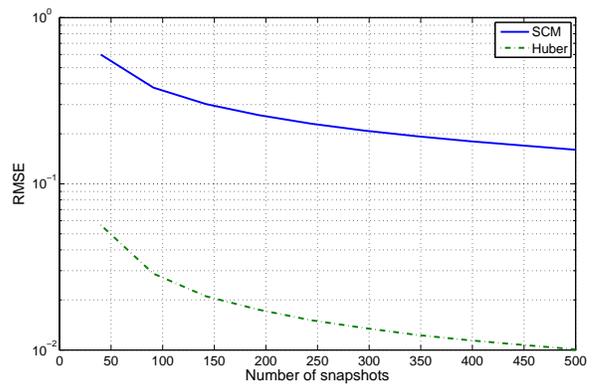}
\caption{One source DOA RMSE ($m=3$ antennas) for Huber's estimate and the SCM, for K-distributed additive noise with a shape parameter $\nu=0.1$.}
\label{fig4}
\end{center}
\end{figure}
\subsection{Asymptotic performance of the ANMF test with the SCM and Huber's $M$-estimator}
Let us give a second illustration of theorem \ref{sect:theorem2}. We consider an adaptive radar receiving a vector $\yg$ of length $m$. The estimated covariance matrix of the environment is $\Mc$ and we try to detect signals of steering vector $\pg$. This steering vector defines the DOA and speed of the target, using the Doppler frequency. The ANMF test statistics \cite{Conte02b} is
\begin{equation}
\Lambda(\Mc|\yg)=\cfrac{|\pg^H\Mc^{-1}\yg|^2}{(\pg^H\Mc^{-1}\yg)(\yg^H\Mc^{-1}\yg)}.
\end{equation}

Firstly, we have considered a Gaussian context and computed $\Lambda(\Mc|\yg)$. In figure \ref{fig5} the vertical scale represents the variance of $\Lambda$ obtained with the SCM and the complex analogue of Huber's $M$-estimator defined in (\ref{eq:VN}). The horizontal scale represents the number of samples used to estimate the covariance matrix. A third curve represents the variance of $\Lambda$ for $\sigma_1 N$ data. As one can see, it overlaps the SCM's curve, illustrating theorem  \ref{sect:theorem2}. The coefficient $\sigma_1$ is equal to $1.067$.
\begin{figure}[H]
\begin{center}
\includegraphics[width=1\linewidth]{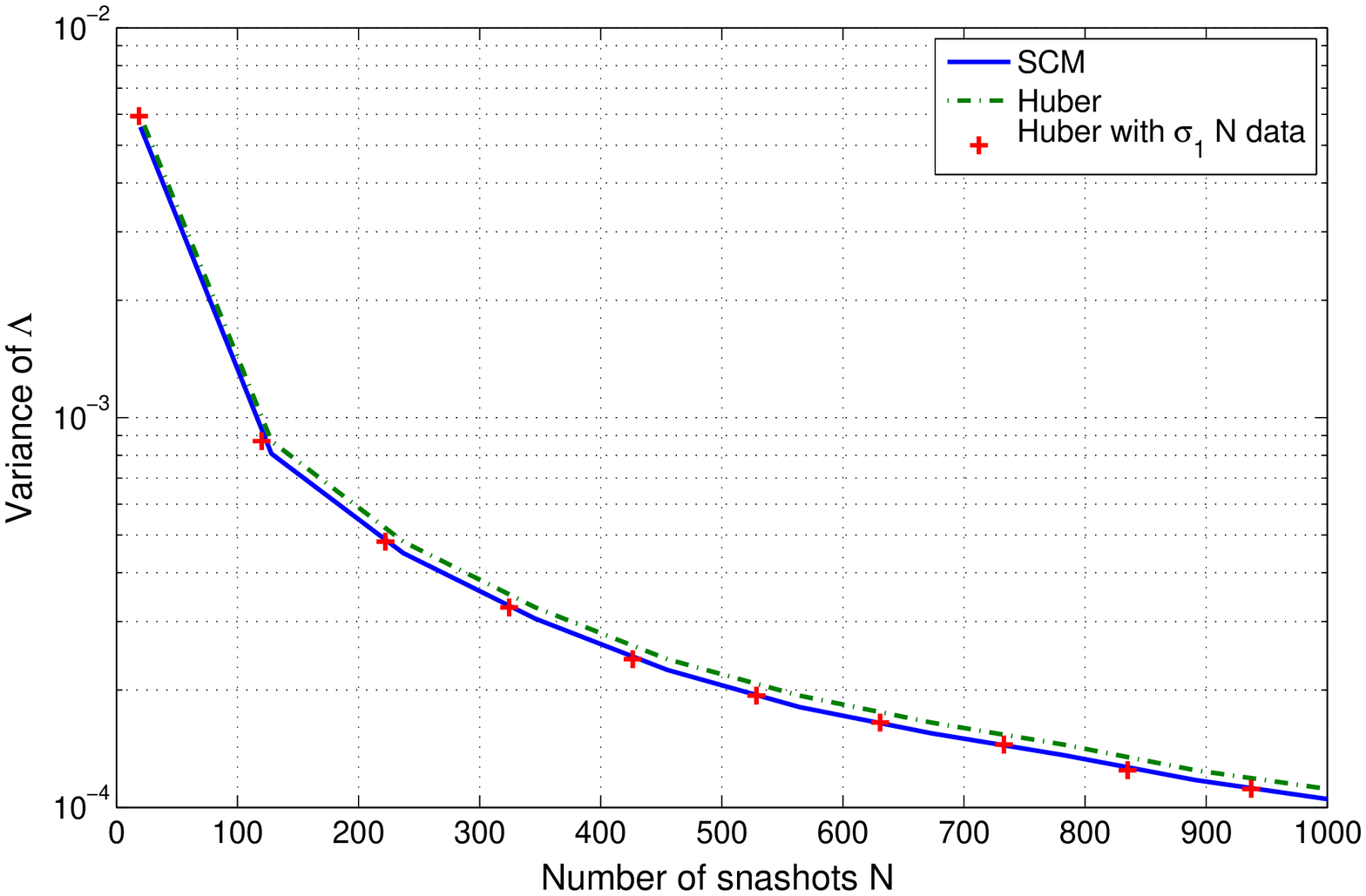}
\caption{Variance on the ANMF detector for Huber's estimate and the SCM estimate, with spatially white Gaussian additive noise.}
\label{fig5}
\end{center}
\end{figure}

Secondly, we have considered a K-distributed environment, with shape parameter firstly equal to 0.1 and then 0.01 for  a more impulsive noise. The figure \ref{fig6} which scales are the same as in figure \ref{fig5}, brings once again to our minds that the SCM is not robust in a non-Gaussian context contrary to Huber's $M$-estimator. Indeed, the more the noise differs from a Gaussian noise, the more the detector's variance is deteriorated in that case while it still gives good results with Huber's $M$-estimator.
\begin{figure}[H]
\begin{center}
\includegraphics[width=1\linewidth]{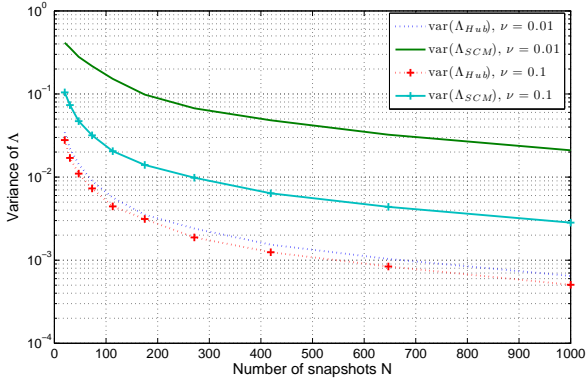}
\caption{Variance on the ANMF detector with the Huber's estimate and the SCM, for K-distributed additive noise with various shape parameters $\nu=0.01$ and $nu=0.1$.} 
\label{fig6}
\end{center}
\end{figure}
\section{Conclusion}
\label{sec:conclu}
In this paper we have analyzed the statistical properties of complex $M$-estimators of the scatter matrix in the framework of complex elliptically distributed data. Firstly, using existing results for real $M$-estimators, we have derived the asymptotic covariance in the complex case. Simulations have checked that when the number of samples increases, the $M$-estimator covariance tends to its theoretical asymptotic value. Secondly, we have extended an interesting property of real $M$-estimators to the complex case. This property states that the asymptotic distributions of any homogeneous function of degree zero of $M$-estimates and Wishart matrices, are the same up to a scale factor. This result has many potential applications in performance analysis of array processing algorithms based on $M$-estimates of the covariance matrix. 
%zero Gaussian distribution, parameters of interest which are function of the covariance matrix, share the same asymptotic ditribution up to a scalar factor $\sigma_1$, whether they are obtained with a complex $M$-estimator or the classical SCM. The function $H$ which gives the parameter must be such as $H(\alpha\Mg)=H(\Mg)$. Roughly speeking, it means that the application considered must be such as the covariance matrix is needed only up to a scale factor, which is often the case in signal processing applications. Thus, in the Gaussian case, $M$-estimators built with $\sigma_1 N$ data achieve the same asymptotic performance as the SCM built with $N$ data and in the non-Gaussian case, they are more robust than the SCM. Moreover, $\sigma_1$ is in most cases very close to 1.
%
% To illustrate our results, we simulate a DOA estimation using the MUSIC method in a Gaussian context. Then, still using the MUSIC method, we compare the results obtained with the classical SCM and a $M$-estimator (the complex equivalent of Huber's $M$-estimator) in a non Gaussian context, bringing to mind that in that case, the SCM gives poor results, whereas $M$-estimators can still give good results. As as second example, we also simulate the ANMF test first in a Gaussian case and then in a non Gaussian case. This simulation leads to the same conclusions as the former example.
%To sum up, this theoretical analysis is in favor of $M$-estimators since they provide robustness and behave almost like standard tools.

\appendix

\begin{proof}{Lemma \ref{lem1}}
\label{appendix}
\subsection{Asymptotic behavior of $\Mc_u$ and $\Mc_v$ }
Let us set
\begin{itemize}
\item $\Mg_u=\Mg_u^{1/2}\Mg_u^{1/2}$,
\item $\Rc_u=\Mg_u^{-1/2}\Mc_u\Mg_u^{-1/2}$ and
\item $\kg_n=\Mg_u^{-1/2}\ug_n$.
\end{itemize}
Since $\Rc_u$ is a consistent estimate of $\Ig_m$, when $N\longrightarrow\infty$, we have $\Rc_u= \Ig_{2m}+\Delta \Rg_u$, considering $\Delta \Rg_u$ small. Thus we have
\begin{eqnarray}
 u\be\ug_n^T\Mc_u^{-1}\ug_n\en&=& u\be\kg_n^T\Rc_u^{-1}\kg_n\en\nonumber\\
  &=&  u\be\kg_n^T(\Ig_{2m}+\Delta \Rg_u)^{-1}\kg_n\en.\nonumber
\end{eqnarray}

A first order expansion of $\Rc_u$ gives $(\Ig_{2m}+\Delta \Rg_u)^{-1}\thickapprox\Ig_{2m}-\Delta \Rg_u$ which leads to
\begin{eqnarray}
  u\be\kg_n^T\Rc_u^{-1}\kg_n\en&=&u\be\|\kg_n\|^2 - \kg_n^T\Delta\Rg_u\kg_n\en\\
&=&  u\be\|\kg_n\|^2\en - u'\be\|\kg_n\|^2 \en \kg_n^T\Delta\Rg_u\kg_n\nonumber\\
&=& a_n+b_n \kg_n^T\Delta\Rg_u\kg_n,\nonumber
\end{eqnarray}
with $a_n=u\be\|\kg_n\|^2\en$ and $b_n= - u'\be\|\kg_n\|^2\en$.
From equation (\ref{eq:Mcu}), we obtain
\begin{equation}
 \Ig_{2m}+\Delta \Rg_u=\cfrac{1}{N}\sum_{n=1}^N a_n\kg_n\kg_n^T+ \cfrac{1}{N}\sum_{n=1}^N b_n \be\kg_n^T\Delta\Rg_u\kg_n\en\kg_n\kg_n^T.\nonumber
\end{equation}
Since $\vec(\kg_n\kg_n^T)=(\kg_n\otimes\kg_n)$ and $\kg_n^T\Delta\Rg_u\kg_n=\vec\be\kg_n^T\Delta\Rg_u\kg_n\en=(\kg_n\otimes\kg_n)^T\vec(\Delta\Rg_u)$, one has the following equation:
\begin{multline}
\vec(\Ig_{2m})+\vec(\Delta \Rg_u)=\cfrac{1}{N}\sum_{n=1}^N a_n\kg_n\otimes\kg_n\\+ \cfrac{1}{N}\sum_{n=1}^N b_n (\kg_n\otimes\kg_n)^T\vec(\Delta\Rg_u)(\kg_n\otimes\kg_n).
\end{multline}

This leads to
\begin{multline}
\be\Ig_{4m^2}-\cfrac{1}{N}\sum_{n=1}^N b_n(\kg_n\otimes\kg_n)(\kg_n\otimes\kg_n)^T\en \vec(\Delta\Rg_u)\\=\cfrac{1}{N}\sum_{n=1}^N a_n(\kg_n\otimes\kg_n )- \vec(\Ig_{2m}).
\end{multline}

Let us denote \\ $\alphag_N=\be\Ig_{4m^2}-\cfrac{1}{N}\disp\sum_{n=1}^N b_n(\kg_n\otimes\kg_n)(\kg_n\otimes\kg_n)^T\en$. Then the previous equation is equivalent to
\begin{equation}\label{eq:dmu}
\vec(\Delta\Rg_u)=\alphag_N^{-1}\be\cfrac{1}{N}\sum_{n=1}^N a_n(\kg_n\otimes\kg_n )- \vec(\Ig_{2m})\en.
\end{equation}
In (\ref{eq:dmu}) we have
\begin{itemize}
\item $\alphag_N \overset{a.s}{\rightarrow} \alphag=\be\Ig_{4m^2}-E[b\,(\kg\otimes\kg)(\kg\otimes\kg)^T]\en$ where $\kg\sim \E(\mathbf{0_{2m,1}}, \sigma_r \Ig_{2m},g_{\zg})$ and $b= -u'\be\|\kg\|^2\en$
\item $\cfrac{1}{N}\disp\sum_{n=1}^N a_n(\kg_n\otimes\kg_n )\overset{a.s}{\rightarrow}  E\left[a\, \kg\otimes\kg\right]=E\left[u\be\|\kg\|^2\|\en \kg\otimes\kg\right]=\vec\be E\left[u\be\|\kg\|^2\|\en \kg\kg^T\right]\en=\vec\be \Ig_{2m}\en$ using (\ref{eq:VsL}) with $\Mg_u$ and replacing $\ug_n$ by $\Mg_u^{1/2}\kg_n$.
\end{itemize}

Now let us denote $\wg_N=\sqrt{N}\be\cfrac{1}{N}\disp\sum_{n=1}^N a_n(\kg_n\otimes\kg_n )- \vec(\Ig_{2m})\en$. We have $\wg_N\overset{d}{\longrightarrow} \wg$, where $\wg$ follows a zero mean Gaussian distribution.

Consequently, the Slutsky theorem gives
\begin{equation}
\sqrt{N}\vec(\Delta\Rg_u)=\alphag_N^{-1}\wg_N\overset{d}{\longrightarrow}\alphag^{-1}\wg.
\label{eq:dru}
\end{equation}
Moreover, one can notice that
\begin{eqnarray*}(\Mg_u^{1/2}\otimes\Mg_u^{1/2})\vec(\Delta\Rg_u) &=&\vec(\Mg_u^{1/2}\Delta\Rg_u\Mg_u^{1/2}) \\ &=&\vec(\Mc_u-\Mg_u),
\end{eqnarray*}

which gives, taking into account equation (\ref{eq:dru}),
\begin{equation}\label{eq:result-Mu}
\sqrt{N}\vec(\Mc_u-\Mg_u)\overset{d}{\longrightarrow}(\Mg_u^{1/2}\otimes\Mg_u^{1/2})\alphag^{-1}\wg.
\end{equation}
Using equation (\ref{eq:Mv=TMuT}), we also have
\begin{multline}\label{eq:result-Mu}
\sqrt{N}\vec(\Mc_v-\Mg_v)=\sqrt{N}(\Trg\otimes\Trg)\vec(\Mc_u-\Mg_u)\\\overset{d}{\longrightarrow}(\Trg\otimes\Trg)(\Mg_u^{1/2}\otimes\Mg_u^{1/2})\alphag^{-1}\wg.
\end{multline}
where $\Trg = \begin{pmatrix}
\mathbf{0}_{m,m} & -\Ig_m \\
\Ig_m& \mathbf{0}_{m,m}
\end{pmatrix}
$.

\subsection{Asymptotic behavior of $\Mc_R$}
Let us denote $\Rc_R=\Mg_R^{-1/2}\Mc_R\Mg_R^{-1/2}$. Since $\Mg_R=\Mg_u$, one has $\kg_n=\Mg_R^{-1/2}\ug_n$.

For all matrices of the form $f(\Ag)$, $\Trg f(\Ag)=f(\Ag)\Trg$. Therefore, since $\Mg_R=f(\Mg_c)$, $\Trg\Mg_R=\Mg_R\Trg$. One has $\Mg_R=\Trg^T \Mg_R \Trg=\be\Trg^T \Mg_R^{1/2} \Trg\en \be\Trg^T \Mg_R^{1/2} \Trg \en$. Therefore $\Mg_R^{1/2}=\Trg^T \Mg_R^{1/2} \Trg$ and $\Trg\Mg_R^{1/2}=\Mg_R^{1/2}\Trg$. This leads to $\Trg\,\kg_n=\Trg\Mg_R^{-1/2}\ug_n=\Mg_R^{-1/2}\vg_n$.

 When $N\longrightarrow\infty$, since $\Rc_R$ is a consistent estimate of $\Ig_{2m}$, $\Rc_R=\Ig_{2m}+\Delta \Rg_R$, with  $\Delta \Rg_R$ small. 
 Similarly to the first part of the proof on has,
\begin{eqnarray*}
u\be\ug_n^T\Mc_R^{-1}\ug_n\en&=& u\be\vg_n^T\Mc_R^{-1}\vg_n\en\\&=&u\be\kg_n^T\Rc_R^{-1}\kg_n\en\\&=& a_n+b_n \kg_n^T\Delta\Rg_R\kg_n.
 \end{eqnarray*}
  
Thus, deriving from equation (\ref{eq:MR2}) we obtain
\begin{eqnarray*}
&&\Ig_{2m}+\Delta \Rg_R \\&&=\cfrac{1}{2N}\sum_{n=1}^N a_n\kg_n\kg_n^T+\cfrac{1}{2N}\sum_{n=1}^N a_n\,\Trg\kg_n\kg_n^T\Trg^T\\&&+ \cfrac{1}{2N}\sum_{n=1}^N b_n \be\kg_n^T\Delta\Rg_R\kg_n\en\kg_n\kg_n^T\\&&+ \cfrac{1}{2N}\sum_{n=1}^N b_n \be\kg_n^T\Trg^T\Delta\Rg_R\Trg\kg_n\en\Trg\kg_n\kg_n^T\Trg^T.
\end{eqnarray*}
Then using the $\vec$ operator, this equation leads to
\begin{eqnarray*}
&& \vec(\Ig_{2m})+\vec(\Delta \Rg_R)\\&&=\cfrac{1}{2N}\sum_{n=1}^N a_n\kg_n\otimes\kg_n +\cfrac{1}{2N}\sum_{n=1}^N a_n\Trg\kg_n\otimes\Trg\kg_n\\&&+ \cfrac{1}{2N}\sum_{n=1}^N b_n  (\kg_n\otimes\kg_n)^T\vec(\Delta\Rg_R) (\kg_n\otimes\kg_n)\\
&&+ \cfrac{1}{2N}\sum_{n=1}^N b_n (\Trg\kg_n\otimes\Trg\kg_n)^T\vec(\Delta\Rg_R) (\Trg\kg_n\otimes\Trg\kg_n).\\
\end{eqnarray*}
This is equivalent to
\begin{eqnarray*}
&&\be\Ig_{4m^2}-\cfrac{1}{2N}\sum_{n=1}^N b_n(\kg_n\otimes\kg_n)(\kg_n\otimes\kg_n)^T-\right.\\&&\left.\cfrac{1}{2N}\sum_{n=1}^N b_n (\Trg\otimes\Trg)(\kg_n\otimes\kg_n)(\kg_n\otimes\kg_n)^T(\Trg\otimes\Trg)^T\en \\ &&\vec(\Delta\Rg_R)\\
&&=\cfrac{1}{2N}\sum_{n=1}^N a_n(\kg_n\otimes\kg_n )\\&&+\cfrac{1}{2N}\sum_{n=1}^N a_n(\Trg\otimes\Trg) (\kg_n\otimes\kg_n )- \vec(\Ig_{2m}).
\end{eqnarray*}

which leads to
\begin{eqnarray*}
&&\vec(\Delta\Rg_R)=\alphagt_N^{-1}\be\cfrac{1}{2N}\sum_{n=1}^N a_n(\kg_n\otimes\kg_n )\right.\\&&\left.+\cfrac{1}{2N}\sum_{n=1}^N a_n(\Trg\otimes\Trg)(\kg_n\otimes\kg_n )- \vec(\Ig_{2m})\en,
\end{eqnarray*}
where $\alphagt_N=\be\Ig_{4m^2}-\cfrac{1}{2N}\disp\sum_{n=1}^N b_n(\kg_n\otimes\kg_n)(\kg_n\otimes\kg_n)^T\right.$\\$\left.-\cfrac{1}{2N}\sum_{n=1}^N b_n(\Trg\otimes\Trg)(\kg_n\otimes\kg_n)(\kg_n\otimes\kg_n)^T(\Trg\otimes\Trg)^T\en$.

Using previous notation $\wg_N$, we obtain
\begin{eqnarray*}
\vec(\Delta\Rg_R) &=&\cfrac{1}{2\sqrt{N}} \,\alphagt_N^{-1}\be \Ig_{4m^2}+(\Trg\otimes\Trg)\en\wg_N
\end{eqnarray*}

One can notice that $ \alphagt_N\overset{a.s}{\longrightarrow}\alphag$ since the $\Trg\kg_n$ have the same distribution as the $\kg_n$. Moreover $\alphagt_N^{-1}\,(\Trg\otimes\Trg)=(\Trg\otimes\Trg)\,\alphagt_N^{-1}$ and  $(\Mg_u^{1/2}\otimes\Mg_u^{1/2})(\Trg\otimes\Trg)=(\Trg\otimes\Trg)(\Mg_u^{1/2}\otimes\Mg_u^{1/2})$. Therefore we obtain

\begin{eqnarray*}
\sqrt{N}\vec(\Mc_R-\Mg_R)\overset{d}{\longrightarrow}\cfrac{1}{2} \be(\Mg_u^{1/2}\otimes\Mg_u^{1/2})\alphag^{-1} \right.\\\left.+(\Trg\otimes\Trg)(\Mg_u^{1/2}\otimes\Mg_u^{1/2})\alphag^{-1}\wg\en.
\end{eqnarray*}
This leads to the conclusion that $\Mc_R$ shares the same asymptotic distribution as  $\cfrac{1}{2}\be\Mc_u+\Mc_v\en$.
\end{proof}

\bibliographystyle{IEEEtran} 
\bibliography{biblio_these}

% Generated by IEEEtran.bst, version: 1.12 (2007/01/11)
\begin{thebibliography}{10}
\providecommand{\url}[1]{#1}
\csname url@samestyle\endcsname
\providecommand{\newblock}{\relax}
\providecommand{\bibinfo}[2]{#2}
\providecommand{\BIBentrySTDinterwordspacing}{\spaceskip=0pt\relax}
\providecommand{\BIBentryALTinterwordstretchfactor}{4}
\providecommand{\BIBentryALTinterwordspacing}{\spaceskip=\fontdimen2\font plus
\BIBentryALTinterwordstretchfactor\fontdimen3\font minus
  \fontdimen4\font\relax}
\providecommand{\BIBforeignlanguage}[2]{{%
\expandafter\ifx\csname l@#1\endcsname\relax
\typeout{** WARNING: IEEEtran.bst: No hyphenation pattern has been}%
\typeout{** loaded for the language `#1'. Using the pattern for}%
\typeout{** the default language instead.}%
\else
\language=\csname l@#1\endcsname
\fi
#2}}
\providecommand{\BIBdecl}{\relax}
\BIBdecl

\bibitem{tyler1983robustness}
D.~Tyler, ``Robustness and efficiency properties of scatter matrices,''
  \emph{Biometrika}, vol.~70, no.~2, p. 411, 1983.

\bibitem{Kay98}
S.~M. Kay, \emph{Fundamentals of Statistical Signal Processing - Detection
  Theory}.\hskip 1em plus 0.5em minus 0.4em\relax Prentice-Hall PTR, 1998,
  vol.~2.

\bibitem{Mahot10}
M.~Mahot, P.~Forster, J.-P. Ovarlez, and F.~Pascal, ``Robustness analysis of
  covariance matrix estimates,'' \emph{European Signal Processing Conference
  (EUSIPCO), Aalborg, Denmark}, August 2010.

\bibitem{huber2009robust}
P.~J. Huber and E.~M. Ronchetti, \emph{{Robust statistics}}.\hskip 1em plus
  0.5em minus 0.4em\relax John Wiley \& Sons Inc, 2009.

\bibitem{Hamp86}
F.~R. Hampel, E.~M. Ronchetti, P.~J. Rousseeuw, and W.~A. Stahel, \emph{Robust
  Statistics: The Approach Based on Influence Functions}, ser. Wiley Series in
  Probability and Statistics.\hskip 1em plus 0.5em minus 0.4em\relax John Wiley
  \& Sons, 1986.

\bibitem{Maronna06}
R.~A. Maronna, D.~R. Martin, and J.~V. Yohai, \emph{Robust Statistics: Theory
  and Methods}, ser. Wiley Series in Probability and Statistics.\hskip 1em plus
  0.5em minus 0.4em\relax John Wiley \& Sons, 2006.

\bibitem{huber1964robust}
P.~J. Huber, ``{Robust estimation of a location parameter},'' \emph{The Annals
  of Mathematical Statistics}, vol.~35, no.~1, pp. 73--101, 1964.

\bibitem{Maronna76}
R.~A. Maronna, ``Robust {$M$}-estimators of multivariate location and
  scatter,'' \emph{Annals of Statistics}, vol.~4, no.~1, pp. 51--67, January
  1976.

\bibitem{kelker1970distribution}
D.~Kelker, ``{Distribution theory of spherical distributions and a
  location-scale parameter generalization},'' \emph{Sankhy{\=a}: The Indian
  Journal of Statistics, Series A}, vol.~32, no.~4, pp. 419--430, 1970.

\bibitem{Watts85}
S.~Watts, ``Radar detection prediction in sea clutter using the compound
  {K}-distribution model,'' \emph{IEE Proceeding, Part. F}, vol. 132, no.~7,
  pp. 613--620, December 1985.

\bibitem{conte1991radar}
E.~Conte, M.~Longo, M.~Lops, and S.~Ullo, ``Radar detection of signals with
  unknown parameters in {K}-distributed clutter,'' \emph{Radar and Signal
  Processing, IEE Proceedings F}, vol. 138, no.~2, pp. 131 --138, Apr. 1991.

\bibitem{Gini98}
F.~Gini, M.~V. Greco, A.~Farina, and P.~Lombardo, ``Optimum and mismatched
  detection against {K}-distributed plus {G}aussian clutter,'' \emph{IEEE
  Trans.-AES}, vol.~34, no.~3, pp. 860--876, July 1998.

\bibitem{tyler1987distribution}
D.~Tyler, ``{A distribution-free M-estimator of multivariate scatter},''
  \emph{The Annals of Statistics}, vol.~15, no.~1, pp. 234--251, 1987.

\bibitem{Pascal07}
F.~Pascal, Y.~Chitour, J.-P. Ovarlez, P.~Forster, and P.~Larzabal, ``Covariance
  structure maximum likelihood estimates in compound gaussian noise\,:
  Existence and algorithm analysis,'' \emph{IEEE Trans.-SP}, vol.~56, no.~1,
  pp. 34--48, January 2008.

\bibitem{ollila2012complex}
E.~Ollila, D.~E. Tyler, V.~Koivunen, and H.~V. Poor, ``Complex elliptically
  symmetric distributions: Survey, new results and applications,'' \emph{Signal
  Processing, IEEE Transactions on}, vol.~60, no.~11, pp. 5597 --5625, nov.
  2012.

\bibitem{ollila2009influence}
E.~Ollila and V.~Koivunen, ``Influence function and asymptotic efficiency of
  scatter matrix based array processors: Case mvdr beamformer,'' \emph{Signal
  Processing, IEEE Transactions on}, vol.~57, no.~1, pp. 247--259, 2009.

\bibitem{ollila2003robust}
------, ``{Robust antenna array processing using M-estimators of
  pseudo-covariance},'' in \emph{Proc. 14th IEEE Int. Symp. Personal, Indoor,
  Mobile Radio Commun.(PIMRC)}, 2003, pp. 7--10.

\bibitem{Ollila2003}
------, ``Influence functions for array covariance matrix estimators,''
  \emph{Proc. IEEE Workshop on Statistical Signal Processing (SSP),ST Louis,
  MO}, pp. 445--448, October 2003.

\bibitem{ollila2003ieee}
E.~Ollila, L.~Quattropani, and V.~Koivunen, ``Robust space-time scatter matrix
  estimator for broadband antenna arrays,'' \emph{Vehicular Technology
  Conference, 2003. VTC 2003-Fall. 2003 IEEE 58th}, vol.~1, pp. 55 -- 59 Vol.1,
  oct. 2003.

\bibitem{couillet2012robust}
R.~Couillet, F.~Pascal, and J.~Silverstein, ``Robust m-estimation for array
  processing: A random matrix approach,'' \emph{arXiv preprint
  arXiv:1204.5320}, 2012.

\bibitem{tyler1982radial}
D.~Tyler, ``{Radial estimates and the test for sphericity},''
  \emph{Biometrika}, vol.~69, no.~2, p. 429, 1982.

\bibitem{Kraut01}
S.~Kraut, L.~L. Scharf, and L.~T. Mc~Whorter, ``Adaptive subspace detectors,''
  \emph{IEEE Trans.-SP}, vol.~49, no.~1, pp. 1--16, January 2001.

\bibitem{kraut1999cfar}
S.~Kraut and L.~Scharf, ``The cfar adaptive subspace detector is a
  scale-invariant glrt,'' \emph{Signal Processing, IEEE Transactions on},
  vol.~47, no.~9, pp. 2538--2541, 1999.

\bibitem{van1995multivariate}
A.~Van~den Bos, ``The multivariate complex normal distribution-a
  generalization,'' \emph{Information Theory, IEEE Transactions on}, vol.~41,
  no.~2, pp. 537--539, 1995.

\bibitem{tyler1988some}
D.~Tyler, ``{Some results on the existence, uniqueness, and computation of the
  M-estimates of multivariate location and scatter},'' \emph{SIAM Journal on
  Scientific and Statistical Computing}, vol.~9, p. 354, 1988.

\bibitem{Kent91}
J.~T. Kent and D.~E. Tyler, ``Redescending {$M$}-estimates of multivariate
  location and scatter,'' \emph{Annals of Statistics}, vol.~19, no.~4, pp.
  2102--2119, December 1991.

\bibitem{bilodeau1999theory}
M.~Bilodeau and D.~Brenner, \emph{{Theory of multivariate statistics}}.\hskip
  1em plus 0.5em minus 0.4em\relax Springer Verlag, 1999.

\bibitem{Pascal07b}
F.~Pascal, P.~Forster, J.-P. Ovarlez, and P.~Larzabal, ``Performance analysis
  of covariance matrix estimates in impulsive noise,'' \emph{IEEE Trans.-SP},
  vol.~56, no.~6, pp. 2206--2217, June 2008.

\bibitem{Conte02b}
E.~Conte, A.~De~Maio, and G.~Ricci, ``Covariance matrix estimation for adaptive
  cfar detection in compound-gaussian clutter,'' \emph{IEEE Trans.-AES},
  vol.~38, no.~2, pp. 415--426, April 2002.

\end{thebibliography}
\end{document}